\documentclass[runningheads]{llncs}

\usepackage[T1]{fontenc}

\usepackage{cite}
\usepackage{booktabs} 
\usepackage{amsmath}
\usepackage{hyperref}
\usepackage{amssymb}
\usepackage{breqn}
\usepackage{algorithm}
\usepackage{algpseudocode}

\newtheorem{assumption}{Assumption}

\usepackage{pgffor}
\usepackage{pgfplots}

\DeclareMathOperator*{\argmaxA}{arg\,max}
\algnewcommand\algorithmicforeach{\textbf{for each}}
\algdef{S}[FOR]{ForEach}[1]{\algorithmicforeach\ #1\ \algorithmicdo}

\renewcommand\algorithmicdo{}

\title{On Stability  and  Learning of Competitive Equilibrium in Generalized Fisher Market Models: A Variational Inequality Approach}
\titlerunning{On Stability  and  Learning of CE in Generalized Fisher Market Models}

\author{Mandar Datar\inst{1}  }
\authorrunning{Mandar Datar}
\institute{CEA-Leti, Universite Grenoble Alpes, F-38000 Grenoble, France\\
\email{mandar.datar@cea.fr}
 \footnote{Author was also associated with Inria Sophia Antipolis and the University of Avignon, France, during the time of this work}
}

\begin{document}


\maketitle

\begin{abstract}
In this work, we study a generalized Fisher market model that incorporates social influence. In this extended model, a buyer's utility depends not only on their own resource allocation but also on the allocations received by their competitors. We propose a novel competitive equilibrium formulation for this generalized Fisher market using a variational inequality approach. This framework effectively captures competitive equilibrium in markets that extend beyond the traditional assumption of homogeneous utility functions. We analyze key structural properties of the proposed variational inequality problem, including monotonicity, stability, and uniqueness. Additionally, we present two decentralized learning algorithms for buyers to achieve competitive equilibrium: a two-timescale stochastic approximation-based tâtonnement method and a trading-post mechanism-based learning method. Finally, we validate the proposed algorithms through numerical simulations.
\end{abstract}



\section{Introduction}
Markets have historically functioned as fundamental mechanisms for allocating resources since the dawn of human civilization. From meticulously recorded commodity prices in ancient Babylon to modern-day
global exchanges, markets facilitate the exchange of goods through intricate pricing systems.
The
fundamental principles governing these exchanges, encapsulated in the notion of equilibrium, have been the focal point of economic inquiry for centuries. In the late 19th century, the French economist Léon Walras established the foundations of modern market theory with the introduction of the concept of a Walrasian equilibrium, also known as a competitive equilibrium (CE) \cite{walras1896elements}. In informal terms, an equilibrium can be defined as a set of prices at which the supply and demand of all goods in the market are in balance.

\par Despite Walras’ foundational contributions, the precise conditions
guaranteeing the existence of such prices remained unresolved until Arrow's groundbreaking work \cite{arrow1954existence}. By leveraging Kakutani’s fixed point theorem, Arrow and Debreu not only settled
the question of existence for a broad class of economic models but also established a fundamental basis for what is now known as general equilibrium theory \cite{McKenzie_CE}.  
However, despite the significance of their result, the non-constructive nature of the proof leaves the challenge of computing equilibrium prices largely unaddressed. 
Indeed, the pursuit of computing equilibrium prices has been a long-standing and significant area of research within the field of economics. Walras pioneered a decentralized price-adjustment process termed \textit{tâtonnement} aimed at reflecting real market dynamics and conjecturing its convergence to equilibrium prices.
 Initially, optimism surrounding this concept was reinforced by Arrow et al. in \cite{arrow1958stability,arrow1959stability}, who illustrated convergence in continuous tâtonnement procedures under specific market conditions, notably the weak gross substitutes (WGS) property. However, the influential counterexample presented in Scarff's study \cite{Scarf_instability} undermined the assumption of universal convergence, underscoring the complex challenges inherent in equilibrium computation. 
\par 
The study of computational complexity in economic settings has provided profound insights into the challenges inherent in finding equilibrium solutions. It has been established that the task of computing even an approximate competitive equilibrium is PPAD-hard \cite{chen2009settling_PPAD}, thereby reinforcing the notion that equilibrium attainment in economic systems is inherently challenging. However, amidst these challenges, Fisher markets \cite{brainard2005compute} have emerged as a notable exception. In contrast to the broader Arrow-Debreu framework, Fisher markets represent a specialized context where competitive equilibria can be efficiently computed under certain conditions. In Fisher markets, there are no firms; instead, buyers possess a single type of commodity, which serves either as an artificial currency or as a fixed endowment defining their budget constraints within the market.
\par Eisenberg and Gale, in their studies \cite{eisenberg1959consensus}\cite{eisenberg1961aggregation}, and the subsequent generalization of their work by Jain et al. in  \cite{jain2007eisenberg}, demonstrated that if the utilities of buyers in the market are continuous, concave and homogeneous of degree one\footnote{A function is called as a homogeneous function of any degree ‘k’ if; when each of its elements is
multiplied by any number $t > 0$; then the value of the function is multiplied by $t^{k}$ 
.}(CCH), the market equilibrium can be determined by solving a convex optimization problem, commonly referred to as the Eisenberg-Gale (EG) program. While the EG program offers a centralized solution for finding equilibrium, it does not accurately represent real-world markets where agents interact with each other. Consequently, the algorithmic game theory community has been keen on developing algorithms that more realistically describe markets and their equilibrium concepts. Over the past two decades, considerable effort has been directed towards designing polynomial-time algorithms that enable buyers to compute competitive equilibria in Fisher markets in a decentralized manner. For example,\cite{Algo2_birnbaum2011distributed},\cite{Algo3_cheung2013tatonnement},\cite{Algo4_avigdor2014convergence},\cite{Algo5_nesterov2018computation},\cite{Algo6_Tracing_equilibrium},\cite{Algo7_cheung2018dynamics},\cite{Algo9_Proportional_dynamics_in_exchange},\cite{Algo10_Greed_leads_to_chaos},\cite{Algo11_goktas2023t},\cite{Algo12_nan2023fast}. Furthermore, recent studies, including \cite{datar2023fisher}, \cite{ICC_Naresh}, and \cite{Fisher_Fog}, have highlighted practical applications of the Fisher market model, demonstrating its growing relevance in real-world contexts. However, most of these approaches rely on the  EG-program and are thus limited to CCH markets.
\par In recent years, there has been a growing interest in generalizing the Fisher market model through the introduction of several innovative variants and the development of the corresponding computational methods. Gao et al. in \cite{gao2023infinite} introduced the concept of infinite-dimensional Fisher markets, expanding the traditional Fisher market model to encompass equilibrium for a continuum of items. In \cite{gao2021online}, Gao et al. and in \cite{jalota2022stochastic}, Jalota et al. studied an online variant of the Fisher market setting, in which users arrive sequentially with privately known utilities and budgets. Zhao et al. in \cite{zhao2023fisher} introduced the Fisher market with social influence, incorporating social dynamics into the traditional Fisher market model. This variant considers how buyers' preferences and decisions are shaped by the behavior and choices of other market participants. Prior to this, Datar et al. in \cite{Strategic_Resource_Pricing} explored competition in wireless communication markets, implicitly linking their analysis to a Fisher market with social influence, though without formalizing it. Despite these advancements, a key limitation of these extended models is that they are all rooted in the EG-program and are typically constrained to cases where buyers have CCH utility functions.
\par In this work, we aim to transcend these limitations by investigating methods to characterize market equilibrium in Fisher markets beyond homogeneous utility functions and by developing algorithms to achieve these equilibria.
To achieve this objective, we revisit the foundational work of Arrow and Debreu \cite{debreu1952social}, who pioneered the concept of generalized Nash equilibrium problem (GNEP)\footnote{ In their seminal work, Arrow and Debreu introduced the concept now widely recognized as the generalized Nash equilibrium problem (GNEP), originally termed the social equilibrium problem. Over time, this concept has been identified by various names, such as pseudo-game, equilibrium programming, coupled constraint equilibrium problem, and abstract economy, depending on its application domain. In this paper, we adopt the term GNEP to denote this problem due to its broad applicability across different contexts.} or games to demonstrate the existence of competitive equilibrium in a general economy model. We begin by modeling the Fisher market as a GNEP, thereby demonstrating the existence of a commutative equilibrium.  However, solving a GNEP is inherently complex. Thus, to make it computationally tractable, we reframe it as a
variational inequality problem. This reformulation simplifies the problem by focusing on the Karush-Kuhn-Tucker (KKT) conditions at equilibrium, which can be expressed as variational inequalities. This approach allows us to develop polynomial-time algorithms capable of efficiently computing the equilibrium.

\section{Our Contributions} We 
model the market equilibrium problem in the generalized Fisher market as the generalized Nash equilibrium problem and develop a novel variational inequality (VI) problem based on the KKT conditions derived from the best responses of buyers at the Nash equilibrium (NE) of the (GNEP). Our analysis demonstrates that the solution to the proposed (VI) problem delineates the competitive equilibrium of the Fisher market. Furthermore, we investigate key properties of variational equilibrium to the VI problem, including its monotonicity, stability, and uniqueness.
\par 
We establish a connection between the variational inequality problem and the prior literature on Fisher markets where buyers exhibit CCH utility functions. Specifically, we demonstrate that the well-known Eisenberg-Gale optimization program in non-social-influential Fisher markets and the buyers-auctioneer game in social-influential Fisher markets \cite{zhao2023fisher} are special cases of the variational inequality problem.
\par We present two decentralized algorithms for computing the competitive equilibrium in Fisher markets, employing two well-known approaches: the \textit{tâtonnement} process and the buyer-centric trading post mechanism\footnote{ The same mechanism has been called by different names in different application domains, for example, the Kelly mechanism \cite{kelly1997charging} in computer networks, and the proportional share scheme by \cite{feldman2008proportional}} \cite{tradingpost}. Our first algorithm introduces a novel variation compared to the traditional \textit{tâtonnement} approach by integrating the concept of two-time-scale stochastic approximation techniques. For the second algorithm, our analysis demonstrates that the trading post-mechanism-based learning scheme closely mirrors the discrete replicator.
Leveraging stochastic approximation techniques, we establish that the proposed algorithm converges effectively in polynomial time. Finally, to validate the proposed algorithms, we conducted numerical simulations aimed at computing the competitive equilibrium in a Fisher market scenario where buyers exhibit non-homogeneous utility functions.

\section{PRELIMINARIES}
\subsection{A Generalized Nash Equilibrium Problem (GNEP) } \label{Preleminier_append}
A Generalized Nash Equilibrium Problem or Game is characterized by a tuple\\ $\mathcal{G} := (\mathcal{N}, (\mathcal{A}_n)_{n \in \mathcal{N}}, (\mathcal{X}_n)_{n \in \mathcal{N}}, (\phi_n)_{n \in \mathcal{N}})$. Here, $\mathcal{N}$ represents a set of players $\left \{  1, \dots, N  \right \}$, where each player $n \in \mathcal{N}$ governs their action $\mathbf{a}_n \in \mathcal{A}_n \subseteq \mathbb{R}^{d_n}$. The joint action space of the players, denoted as $\mathcal{A} = \prod_{n \in \mathcal{N}} \mathcal{A}_n$, encompasses all possible combinations of individual actions. Each player $n$ is associated with an objective function $\phi_n: \mathcal{A} \rightarrow \mathbb{R}$, which is concave and continuous in $\mathbf{a}_{n}$ and relies on both their own action $\mathbf{a}_{n}$ and the actions $\mathbf{a}_{-n}$ of all other players. This dependence is represented as $\phi_n(\mathbf{a}_n, \mathbf{a}_{-n})$. Given a particular $ \mathbf{a}_{-n} \in \mathcal{A}_{-n}=\prod_{m \in \mathcal{N},  m \neq n} \mathcal{A}_{m} $, each player $n$ endeavors to maximize their objective function by selecting a feasible action such that $\mathbf{a}_n \in \mathcal{X}_{n}(\mathbf{a}_{-n}) \subseteq \mathcal{A}_n$. Here, the feasible action space or the strategy space of each player depends on the actions of others. Where $\mathcal{X}_{n}:\mathcal{A}_{-n}\rightrightarrows \mathcal{A}_{n}$ is a non empty, continuous, convex and  compact map. A generalized Nash equilibrium problem (GNEP) involves $N$ constrained optimization problems. In essence, for each player $n$, it tackles the optimization problem $Q_n(\mathbf{a}_{-n})$, 
\begin{equation*}
 \boxed{  
\begin{aligned}
Q_{n}(\mathbf{a}_{-n})\hspace{20pt}& \underset{\mathbf{a}_{n}}{\text{maximize}}
&&\phi_{n}\left(\mathbf{a}_{n},\mathbf{a}_{-n}\right),\\
& \text{subject to}
&& \mathbf{a}_{n}\in \mathcal{X}_{n}(\mathbf{a}_{-n}) .
\end{aligned}}
\end{equation*}
The generalized Nash equilibrium \( \mathbf{a}^* = (\mathbf{a}^*_1, \ldots, \mathbf{a}^*_N) \) in the generalized game is a state where no player \( n \) desires to unilaterally deviate from their strategy in the equilibrium profile \( \mathbf{a}^* \). Moreover, this equilibrium \( \mathbf{a}^* \) must satisfy the constraints of each agent \( n \in \mathcal{N} \), denoted as \( \mathbf{a}_n^* \in \mathcal{X}_n(\mathbf{a}_{-n}^{*}) \).
\begin{definition}[Genralized Nash equilibrium (GNE)] 
A strategy profile $(\mathbf{a}^{*})$ is a generalized Nash equilibrium (GNE) of the game $\mathcal{G}$ if $a^{*}_n\in SOL(Q_{n}(\mathbf{a}^*_{-n}))\;$ for all  $n \in \mathcal{N}$, \emph{i.e.,}  $\phi_{n}(\mathbf{a}^{*}_{n},\mathbf{a}^{*}_{-n})\geq \phi_{n}(\mathbf{a}_{n},\mathbf{a}^{*}_{-n})$ for all $\mathbf{a}_{n}\in \mathcal{X}_{n}(\mathbf{a}^{*}_{-n})$ and all $n\in \mathcal{N}.$
\end{definition}
In the above definition $SOL(Q_{n}(.))$ represents the solution of the optimization problem
$Q_{n}(.)$. If in a game $\mathcal{G}$, the objective function of each player $n\in \mathcal{N}$ $\phi_{n}()$ is concave and the feasible action set $\mathcal{X}_{n}(\mathbf{a}_{-n})$ is convex and compact, then the existence of GNE is guaranteed by the work of \cite{arrow1954existence}  
\begin{definition}[Variational Inequality (VI) problem]
  Consider a non-empty closed convex set $\mathcal{C} \subset \mathbb{R}^n$ and a mapping $F: \mathcal{C} \rightarrow \mathbb{R}^n$. The Variational Inequality (VI) problem is defined as follows:

\[
\boxed{
\begin{aligned}
\text{VI}(\mathcal{C}, F) \hspace{20pt}
& \text{Find } \mathbf{a}^* \in \mathcal{C} \text{ such that}, \\
& \langle F(\mathbf{a}^*), \mathbf{a} - \mathbf{a}^* \rangle \leq 0, \quad \forall \mathbf{a} \in \mathcal{C}.
\end{aligned}}
\]
   
 \end{definition}
 
Here, $\langle \cdot, \cdot \rangle$ denotes the inner product in $\mathbb{R}^n$. This inequality essentially expresses a maximization condition for the functional $F(\mathbf{a})$ over the set $\mathcal{C}$, ensuring that $\mathbf{a}^*$ is a solution where the gradient $F(\mathbf{a}^*)$ is orthogonal to the tangent space of $\mathcal{C}$ at $\mathbf{a}^*$.

Now, let us introduce the concept of Variational Equilibrium (VE) in GNEP, which arises particularly in scenarios when GNEP involves coupled or shared constraints. Consider a convex compact set denoted as $\mathcal{C} \subset \mathbb{R}^{d}$, where $d=\sum_n d_n$, the sum of dimensions for each player $n$ in-game $\mathcal{G}$ and $\mathcal{C}$ set represent the shared constraints of the game $\mathcal{G}$. Now, for each player $n$, we define $\mathcal{X}_{n}(\mathbf{a}_{-n})$ as the set of feasible strategies.
\begin{equation}
    \mathcal{X}_n(\mathbf{a}_{-n})=\left \{ \mathbf{a}_n\in\mathcal{A}_n \mid (\mathbf{a}_n,\mathbf{a}_{-n})\in\mathcal{C}  \right \}
\end{equation}

\begin{definition}[Variational Equilibrium (VE)] 
A strategy profile $(\mathbf{a}^{*})$ is called the Variational Equilibrium (VE) of the game $\mathcal{G}$ with shared constraints $\mathcal{C}$ if $\mathbf{a}^{*}\in SOL(\text{VI}(\mathcal{C}, F))$, where $F(\mathbf{a}) = \left( \nabla_{\mathbf{a}_n}\phi_{n}, \dots, \nabla_{\mathbf{a}_n}\phi_{N}\right)$ represents the pseudogradient of the utility profiles of the players, and $\nabla_{\mathbf{a}_n}$ denotes the partial derivative with respect to $\mathbf{a}_n$.
\end{definition}

\begin{definition}[Monotone Game]\label{defn_monot}
A game with profiles of strategies $\mathbf{a}$ and profiles of utility functions $\phi$ is called a Monotone, strictly monotone (Diagonally strict concave (DSC)\cite{rosen1965existence}), or Strongly monotone game~ if for every distinct
$\mathbf{a}$ and $\mathbf{a}'$,
\begin{align}
\text{Monotone:}\quad&\langle F(\mathbf{a})-F(\mathbf{a}'), \mathbf{a} - \mathbf{a}' \rangle \leq 0, \label{dsc}\\
\text{Strictly Monotone:}\quad& \langle F(\mathbf{a})-F(\mathbf{a}'), \mathbf{a} - \mathbf{a}' \rangle < 0, \label{dsc}\\
\text{Strongly Monotone:}\quad& \langle F(\mathbf{a})-F(\mathbf{a}'), \mathbf{a} - \mathbf{a}' \rangle < c \left \| \mathbf{a}-\mathbf{a}' \right \| 
\end{align}  
where $c$ is a positive 
with $F$ the concatenation of the gradients of the players' utility functions
\begin{equation}
F(\mathbf{a})=\begin{bmatrix}
\displaystyle \nabla _{1} \phi_1(\mathbf{a}),
\displaystyle \nabla_{2} \phi_2(\mathbf{a}),
\dots, \displaystyle \nabla_{N} \phi_N(\mathbf{a})
\end{bmatrix},
\end{equation}
where $\nabla _{n} \phi_{n}(\mathbf{a})$ denotes the gradient of objective function of player $n$ with respect to his own strategy $\mathbf{a}_{n}$.
\end{definition}

\begin{definition}[No saturation]
   A utility function $U$ satisfies the no-saturation property if
\[
\forall \mathbf{a} \in \mathcal{A}, \, \exists \mathbf{a}' \in \mathcal{A} \text{ such that } U(\mathbf{a}') > U(\mathbf{a}).
\]
\end{definition}
\subsection{Generalized Fisher Market}\label{premelineri}
The generalized Fisher market, defined as $$\mathcal{M}:=\left \langle \mathcal{N},\left(\mathbf{x}_{n}\in\mathbb{R}^K\right)_{n\in\mathcal{N}},\left(U_{n}\right)_{n\in\mathcal{N}},\left(B_{n}\right)_{n\in\mathcal{N}},\mathbf{p}\in\mathbb{R}^K \right \rangle,$$ consists of a set of buyers denoted as $\mathcal{N} = \left \{ 1, 2, \ldots, N \right \}$ who have demands for a set of divisible goods or resources, represented as $\mathcal{K} = \left \{ 1, 2, \ldots, K \right \}$. Each buyer $n$ expresses its resource allocation preferences through a vector $\mathbf{x}_n = (x_{n1}, x_{n2}, \ldots, x_{nK})$, where $x_{nk}$ signifies the quantity of resource $k$ required by buyer $n$. The collective preferences of the resources of all buyers are encapsulated in the vector formed by all these strategies and are represented as $\mathbf{x} := \left(\mathbf{x}_n\right)_{n=1}^{N}$. Each buyer $n$ is associated with a utility function denoted as $U_{n}(\mathbf{x}_n, \mathbf{x}_{-n}): \mathbb{R}^{K\times N} \rightarrow \mathbb{R}_{+}$, indicating the quantified value or utility accrued by the buyer $n$ upon acquiring $\mathbf{x}_{n}$ resources, while its adversaries have been assigned $\mathbf{x}_{-n}:=\left( \mathbf{x}_{m} \right)_{m \neq n}^{N}$. In contrast to the traditional Fisher market, where buyers' utility depends solely on their acquired resources, we examine a broader scenario where buyers' utility depends not only on their resources but also on those of their opponents. \begin{assumption}
For each buyer $n\in\mathcal{N}$, $U_{n}$ is continuous in $\mathbf{x}$, concave in $\mathbf{x}_{n}$, and satisfies no saturation \label{assump_existence_GNE}
\end{assumption} 
Without loss of generality, we assume that each good $k$ has a unit capacity, and each buyer $n$ is endowed with a positive monetary budget $B_n$ such that $\sum_{n} B_{n} = 1$. Let $\mathbf{p}= (p_1,\dots,p_K)$ be a price vector representing the prices for all resources, where $p_k$ denotes the price per unit for the resource $k$. Given the set of prices $ \mathbf{p} \in \mathbb{R}^K_{+}$ for the resources, the feasible demand set of the buyer $n$ is defined as the set of demands that satisfies its budget.
\begin{align}
\mathcal{X}_{n}(\mathbf{p})= \left \{  \mathbf{x}_n \mid  \mathbf{x}_{n}\in \mathbb{R}^{m}, \sum_{k}{x_{nk} p_{k}}= B_{n} \right \}\label{feasibledemand}
\end{align} We assume that buyers exhibit rational, self-interested behaviour, with each buyer seeking to maximize their individual utility. In this context, they consider the decisions made by their peer buyers when determining their own actions, where the decision problem for each buyer $n$ is given as 
\begin{equation}
 \begin{aligned}
Q_{n}(\mathbf{x}_{-n};\mathbf{p})\hspace{20pt}& \underset{\mathbf{x}_{n}}{\text{maximize}}
&&U_{n}\left(\mathbf{x}_{n},\mathbf{x}_{-n}\right),\\
& \text{subject to}
&& \mathbf{x}_{n}\in \mathcal{X}_{n}(\mathbf{p}) .
\end{aligned}
\end{equation}
Given resource prices, buyers strategic interaction gives rise to a non-cooperative game, denoted as $\mathcal{G}(\mathbf{p}) \triangleq \left \langle \mathcal{N},\left(\mathbf{x}_{n}\in\mathbb{R}^K\right)_{n\in\mathcal{N}},\left(\mathcal{X}_{n}(\mathbf{p})\right)_{n\in\mathcal{N}},\left(U_{n}\right)_{n\in\mathcal{N}} \right \rangle$, comprising the set of players $\mathcal{N}$, their respective strategy spaces $\left(\mathcal{X}_{n}(\mathbf{p})\right){n\in\mathcal{N}}$, and their utility functions $\left(U_{n}\right)_{n\in\mathcal{N}}$ and Let $\mathcal{X}(\mathbf{p})=\mathcal{X}_{1}(\mathbf{p})\times\mathcal{X}_2(\mathbf{p})\times\dots \mathcal{X}_N(\mathbf{p}) $ defines the joint action space for the game. Here, we deliberately use the notation $\mathcal{G}(\mathbf{p})$ to emphasize the dependency of the game on the prices of goods, $\mathbf{p}$. The outcome of this game is determined by the Nash equilibrium.

\begin{definition} 
A strategy profile \((\mathbf{x}^{*}) \in \mathcal{X}(\mathbf{p})\) is called the Nash equilibrium (NE) of the game \(\mathcal{G}(\mathbf{p})\) if
\[ U_{n}(\mathbf{x}^{*}_{n}, \mathbf{x}^{*}_{-n}) \geq U_{n}(\mathbf{x}_{n}, \mathbf{x}^{*}_{-n}) \]
for all \(\mathbf{x}_{n} \in \mathcal{X}_{n}(\mathbf{p})\) and for all \(n \in \mathcal{N}\).

\end{definition}
In this work, we operate under the assumption that the game induced through the strategic interactions of buyers is Monotone.
\begin{assumption}
We assume that the game $\mathcal{G}(p)$ exhibits strict monotonicity.\label{assup_stricmono_utilti}
\end{assumption}
Throughout the remainder of this paper, the term Fisher market refers to the generalized Fisher market unless explicitly stated otherwise. The traditional Fisher market, devoid of social influence, can be seen as a special instance within the framework of the generalized Fisher market, where utility functions remain unaffected by opponents' allocations.

\begin{definition}[Demand Function]
    If the game $\mathcal{G}(\mathbf{p})$ has a unique NE for each $\mathbf{p} \in \mathcal{P} \subset \mathbb{R}^{K}_{+} $, this arises a function:
    \begin{align}
        \mathbf{x}^{*}_{n}(\mathbf{p}): & \mathcal{P} \rightarrow \mathbb{R}^{K}_{+} \\
        & \mathbf{p} \mapsto \mathbf{x}^{*}_{n}(\mathbf{p}),
    \end{align}
    which is called the demand function.
\end{definition}

\begin{definition}[Aggregate Excess Demand Function]
    We now define a particular aggregate excess demand function:
    \begin{align}
        z_{k}(\mathbf{p}): & \mathcal{P} \rightarrow \mathbb{R} \quad \forall k \in \mathcal{K} \\
        & \mathbf{p} \mapsto z_{k}(\mathbf{p}) = \left(\sum_{n} x^{*}_{nk}(\mathbf{p}) - 1\right).
    \end{align}
    By grouping these components into a vector, we introduce:
    \begin{equation}
        \mathbf{z}(\mathbf{p}) := \left(z_{1}(\mathbf{p}), \dots, z_{K}(\mathbf{p})\right) \in \mathbb{R}^{K}.
    \end{equation}
\end{definition}

\begin{definition}
A competitive equilibrium for the market $\mathcal{M}$ is defined as a pair of prices and allocation $(\mathbf{p}^*,\mathbf{x}^{*})$, where the market clears its resources and buyers get their favorite resource bundle. Mathematically $(\mathbf{p}^*,\mathbf{x}^{*})$ is CE if the following two conditions are satisfied.

\begin{description}
    \item[C1] Given the resource price vector, every buyer $n$ spend its budget such that it receives resource bundle $\mathbf{x}^*_{n}$ that maximizes its utility. 
\begin{align}
\mathbf{x}_{n}^{*}\in \argmaxA \left \{  U_{n}(\mathbf{x}_n,\mathbf{x}_{-n}^{*}) \mid  \mathbf{x}_{n} \in \mathcal{X}_n(\mathbf{p}) \right \}\;\forall n \in \mathcal{N}\label{bestdemand}
\end{align}\label{C1}
\item[C2] Either the total demand of each resource meets the capacity and will be positively priced; otherwise, that resource  has zero price, i.e., we have:
\begin{align}
p^{*}_k\left ( \sum_{n} x^{*}_{nk}-1 \right )=0,\forall k\in\mathcal{K}
\end{align}\label{c2}
\end{description}
\end{definition}

\section{Competitive equilibrium  Problem}
In this study, our objective is to develop a method to determine competitive equilibrium in a generalized Fisher market. Initially, we illustrate that this equilibrium can be attained through a pseudo-game or GNEP involving both buyers and auctioneer.
\begin{definition}[Auctioneer Buyer Pseudo Game]
We consider the game $\widehat{\mathcal{G}}$ consisting of \textit{ players:} $N$ buyers and one auctioneer.
\textit{Actions:}
Each buyer chooses an allocation $\mathbf{x}_{n} \in \mathcal{A}_n = \mathbb{R}^{K}_{+}$ and the auctioneer chooses prices $\mathbf{p} \in \mathcal{A}_{N+1}=\mathbb{R}^{K}_{+}$.
\textit{Action space:}
The feasible action space for each buyer $n$ is $\mathcal{X}_{n}(\mathbf{p})$ as defined in \eqref{feasibledemand} while for the auctioneer it is $\mathcal{P} = \left\{ \mathbf{p} \mid \mathbf{p} \in \mathbb{R}^{K}_{+}, \sum_{k} p_{k} = 1 \right\}$.
\textit{Utility:}
For buyers, $\phi_{n} = U_n$.
         For the auctioneer, $\phi_{N+1} = \sum_{k \in \mathcal{K}} p_{k} \left( \sum_{n \in \mathcal{N}} x_{nk} - 1 \right)$.
\end{definition}
\begin{theorem}[\cite{zhao2023fisher}]
    For any Fisher market $\mathcal{M}$, where buyers utilities satisfy the Assumption \ref{assump_existence_GNE},  the set of competitive equilibria corresponds precisely to the set of generalized Nash equilibria (GNE) of the associated auctioneer-buyer pseudo-game $\widehat{\mathcal{G}}$.
\end{theorem}
Thus far, we have studied a connection between the CE in a Fisher market and the GNE in an Auctioneer-Buyer pseudo-game, with
its existence being ensured under mild conditions. However, a fundamental challenge remains: How can we effectively compute the GNE in this game? The task of determining a NE in a GNEP is generally recognized as challenging. However,
within a subset of GNEPs where players share common coupled constraints, the problem can be
reformulated as a variational inequality (VI) problem \cite{facchinei2007generalized}, with its solution referred to as a variational equilibrium (VE). This type of equilibrium is well-studied and has efficient
polynomial-time algorithms for solution.  Unfortunately, in our specific scenario, the shared coupled constraint between players,
expressed as $\sum_{k\in\mathcal{K}}\sum_{n\in\mathcal{N}} p_k x_{nk}=1$, does not exhibit joint convexity. Consequently, the game cannot be
directly resolved as a VI problem. However, we introduce a novel formulation of the VI problem, different
from the standard framework, and demonstrate that its solution coincides precisely with 
the game $\widehat{\mathcal{G}}$, \emph{i.e,} CE of the market $\mathcal{M}$. Before exploring the VI formulation, we first examine the KKT conditions at the GNE of the game $\widehat{\mathcal{G}}$, which play a pivotal role in obtaining our main result.

\begin{proposition}
$(\mathbf{p}^*,\mathbf{x}^*)$ is CE to market $\mathcal{M}$ (GNE of game $\widehat{\mathcal{G}}$)  if it satisfies the system of KKT conditions \eqref{KKT_comb} \label{proposition_KKT}  

\begin{subequations}

\begin{align}
\;&\text{For each buyer $n$ in }\mathcal{N}\nonumber\\
\;& \left[ \frac{B_n\displaystyle \nabla_{nk} U_{n}(\mathbf{x}_n,\mathbf{x}_{-n}^{*})}{\sum_{k'\in\mathcal{K}}\displaystyle \nabla_{nk'} U_{n}(\mathbf{x}_n,\mathbf{x}_{-n}^{*})x_{nk'}}\right]_{\mathbf{x}_{n}=\mathbf{x}_{n}^{*}}-p_{k}^{*}+\gamma^{*}_{nk}=0,\;&&\forall k \in \mathcal{K}\label{kkt_comb:1}\\
\;&p_{k}^{*}\left(\sum_{n\in\mathcal{N}} x_{nk}^{*}-1\right)=0, \gamma_{nk}^{*}x^{*}_{nk}=0,\; &&\forall k\in\mathcal{K}\label{kkt_comb:2}\\
\;& p_{k}^{*}\geq 0,\gamma^{*}_{nk}\geq 0,\; &&\forall k\in\mathcal{K}\label{kkt_comb:4}
 \end{align}\label{KKT_comb}
\end{subequations}
\end{proposition}
\begin{proof}
Appendix~\ref{append1}
\end{proof}
In the KKT conditions  \eqref{KKT_comb}, the complementary slackness and dual feasibility conditions related to the budget constraint \eqref{feasibledemand} for each buyer do not appear explicitly. However, they are implicitly captured by the stationarity condition \eqref{kkt_comb:1}, which allows us to reformulate these conditions as a variational inequality problem. Interestingly, multiplying both sides of \eqref{kkt_comb:1} by $x_{nk}$ yields:
\begin{equation}
    \left[ \frac{B_n \displaystyle \nabla_{nk} U_{n}(\mathbf{x}_n,\mathbf{x}_{-n}^{*})x_{nk}}{\sum_{k'\in\mathcal{K}}\displaystyle \nabla_{nk'} U_{n}(\mathbf{x}_n,\mathbf{x}_{-n}^{*})x_{nk'}}\right]_{\mathbf{x}_{n}=\mathbf{x}_{n}^{*}} - b_{nk}^{*} = 0,\;\forall k \in \mathcal{K}.
\end{equation} where $b_{nk}^{*}=p_k^{*}\cdot x_{nk}^{*}$. This equation resembles a family of dynamics such as Multiplicative Weights Update \cite{bailey2018multiplicativ} and replicator dynamics\cite{replicator_SorinSylvain}, illustrating how the budget is distributed proportionally to the weighted marginal utilities between resources. We will see later in section \ref{sec_learning} how this structure can be used to design an equilibrium learning algorithm.

\subsection{Variational Inequality (VI) problem formulation of game $\widehat{\mathcal{G}}$}
\begin{definition}
    Consider the Variational Inequality  problem  $\text{VI}(\mathcal{C}, F)$, with $F$ and  the (coupled) constraint $\mathcal{C}$ are defined as \eqref{VI_F} and \eqref{VI_C} respectively   

\begin{equation}
F_{nk}(\mathbf{x}):\mathbb{R}^{NK}\rightarrow \mathbb{R}\;\text{as}\;
   F_{nk}(\mathbf{x})=\frac{B_{n}\displaystyle \nabla_{nk}U_n(\mathbf{x})}{\sum_{k'\in\mathcal{K}}\displaystyle \nabla_{nk'}U_n(\mathbf{x})x_{nk'}}  \;\label{VI_F}
\end{equation}
\begin{equation}
    \text{and}\;F_{n}(\mathbf{x})=(F_{nk}(\mathbf{x}))_{k=1}^{K}\text{and}\;
   F(\mathbf{x})=(F_{n}(\mathbf{x}))_{n=1}^{N}\nonumber
\end{equation}
 \begin{equation}
       \mathcal{C}= \left \{  \mathbf{x} \mid  \mathbf{x}\in \mathbb{R}^{NK},\sum_{n\in\mathcal{N}}x_{nk} \leq 1, \forall k \in \mathcal{K} \right \} \label{VI_C}
   \end{equation}

\end{definition}
\begin{assumption}
$F(\mathbf{x})$  is continuous  and strictly monotone in $\mathbf{x}$\label{assum_varri}
\end{assumption}

\begin{theorem}   
If the utilities of the buyers in market $\mathcal{M}$ satisfy Assumption~\ref{assump_existence_GNE},\ref{assup_stricmono_utilti} and \ref{assum_varri} then 
\begin{itemize}
\item variational equilibrium  $\mathbf{x}_{VE}^*$ to a problem $\text{VI}(\mathcal{C}, F)$ is unique 
    \item  $\mathbf{x}^*$ is CE allocation if and only if it is a variational equilibrium to $\text{VI}(\mathcal{C}, F)$ 
    \item If  $\mathbf{x}_{VE}^*$ is the solution to $\text{VI}(\mathcal{C}, F)$ then corresponding KKT conditions are satisfied with optimal Lagrange multiplier vector $\lambda$ which correspond the CE prices $\mathbf{p}^{*}$.  
\end{itemize}

\end{theorem}

\begin{proof}
Appendix~\ref{append2}
\end{proof}
In the above theorem, the uniqueness of the VE refers only to the allocation $\mathbf{x}_{VE}^*$. Although the allocation is unique, there can be multiple price vectors $\mathbf{p}^*$ that lead to the same allocation. So far, we have demonstrated how the variational inequality problem formulation of the game 
$\widehat{\mathcal{G}}$ can be derived from the KKT conditions obtained at the GNE. In subsequent sections, we establish several existing results from the literature as corollaries to our findings. Specifically, we investigate scenarios where the buyers' utilities exhibit homogeneity of function degree one, illustrating how these cases are special instances within our broader framework.

\section{Homogeneous utiltiy}
Let's delve into a particular market scenario which is extensively explored in the literature, where the utilities of buyers are modelled as CCH. We aim to establish the connection between existing findings in the literature and our results.

\begin{corollary}
If each buyer's utility function within market $\mathcal{M}$ is continuous, concave and homogeneous of degree one (CCH) with respect to their own decisions, then the following results hold:
\begin{itemize}
    \item{Theorem 7. \cite{zhao2023fisher} } The competitive equilibrium pair $(\mathbf{p}^*,\mathbf{x}^*)$ can be determined by solving a variational equilibrium within a coupled constraint game. In this game, each buyer $n \in \mathcal{N}$ seeks to maximize their utility $\phi_n(\mathbf{x}_n,\mathbf{x}_{-n})= B_{n}\log{(U_n(\mathbf{x}))}$ subject to coupled constraints. $\sum_{ n \in \mathcal{N}} x_{nk}    \leq 1,\; \forall k\in {\mathcal{K}}$
    \item  Furthermore, in the non-social influence scenario where each buyer's utility depends only on their own decision, the competitive equilibrium can be computed by solving an (Eisenberg-Gale) optimization program \eqref{EG_prog}. 
\end{itemize}
 \begin{subequations}\label{galeconvex}
\begin{align}
&\hspace{5cm} \underset{\mathbf{x}}{\text{Maximize :}} 
& & \sum_{n\in \mathcal {N}}B_n \log(U_n(\mathbf{x}_{n})) \label{primal}\\
& \hspace{5cm} \text{subject to :} 
& & \sum_{ n \in \mathcal{N}} x_{nk}    \leq 1, \quad \forall k\in {\mathcal{K}} \label{const3_m}.
 \end{align}\label{EG_prog}
 \end{subequations}
\label{extension_corollary}
where $\mathbf{x}^*$ represents the equilibrium allocation, while $\mathbf{p}^*$ denotes the Lagrange multipliers vector associated with the constraints \eqref{const3_m}. 
\end{corollary}
\begin{proof}
Appendix~\ref{append3}
\end{proof}
In the next section, we examine several properties of the variational inequality problem which serve as the basis for developing the subsequent decentralized learning algorithm. 
\section{ 
Monotonicity, Stability, and Uniqueness in $\text{VI}(\mathcal{C}, F)$}
\begin{definition}[Variational stability~\cite{Mertikopoulos2019}]
    We say solution (equilibrium) to $VI(F,\mathcal{C})$ is stable if there 
exists a neighborhood $N_{bd}(\mathbf{x}^*)$ of $\mathbf{x}^*$
such that \begin{equation}
 \boxed{  \langle F(\mathbf{x}), \mathbf{x}-\mathbf{x}^*  \rangle \leq 0 \quad \forall \mathbf{x} \in  N_{bd}(\mathbf{x}^*)}
\end{equation}
with equality if and only if $\mathbf{x}= \mathbf{x}^*$. In particular, if $N_{bd}$ can be taken to be all of $\mathcal{C}$,
we say that $\mathbf{x}^*$ is globally variationally stable (or globally stable for short).
\end{definition}
 
\begin{proposition}
If the function \( F \) in the  problem \(\text{VI}(\mathcal{C}, F)\) satisfies Assumption~\ref{assum_varri}, then the VE for the problem \(\text{VI}(\mathcal{C}, F)\) is unique and globally stable. 
\end{proposition}
\begin{proof}If Assumption  \ref{assum_varri} holds, then the uniqueness follows from Theorem 1.6 \cite{nagurney2009network}. Let $\mathbf{x}^*$ denote the VE   of the problem $\text{VI}(\mathcal{C}, F)$. Given that $F$ is strictly monotone by Assumption~\ref{assum_varri}, we have:
\[
\langle F(\mathbf{x}) - F(\mathbf{x}^*), \mathbf{x} - \mathbf{x}^* \rangle \leq 0
\]

Since $\mathbf{x}^*$ is the VE solution of $\text{VI}(\mathcal{C}, F)$, by definition it satisfies:
\[
\langle F(\mathbf{x}^*), \mathbf{x} - \mathbf{x}^* \rangle \leq 0
\]

Subtracting these two inequalities gives us:
\[
\langle F(\mathbf{x}), \mathbf{x} - \mathbf{x}^* \rangle \leq 0
\]

\end{proof}

\subsection{Test for stability} \label{append_Test_stability}
In this section, we investigate the conditions that guarantee the stability of the variational equilibrium to the VI problem \eqref{VI_C}. Consider $H(\mathbf{x})$ as the block matrix defined as 

\begin{align}
H(\mathbf{x})=\begin{bmatrix}
\begin{bmatrix}
H^{1}_{1}
\end{bmatrix} & \begin{bmatrix}
H^{1}_{2}
\end{bmatrix} & \cdots & \begin{bmatrix}
H^{1}_{N}
\end{bmatrix} \\
\begin{bmatrix}
H^{2}_{1}
\end{bmatrix} & \begin{bmatrix}
H^{2}_{2}
\end{bmatrix} & \cdots & \begin{bmatrix}
H^{2}_{N}
\end{bmatrix} \\
\vdots & \vdots & \ddots & \vdots \\
\begin{bmatrix}
H^{N}_{1}
\end{bmatrix} & \begin{bmatrix}
H^{N}_{2}
\end{bmatrix} & \cdots & \begin{bmatrix}
H^{N}_{N}
\end{bmatrix}
\end{bmatrix}
\text{where} \;
H^{n}_{m}=\begin{bmatrix}
h^{n1}_{m1} & h^{n1}_{m2} & \cdots & h^{n1}_{mK} \\
h^{n2}_{m1} & h^{n2}_{m2} & \cdots & h^{n2}_{mK} \\
\vdots & \vdots & \ddots & \vdots \\
h^{nK}_{m1} & h^{nK}_{m2} & \cdots & h^{nK}_{mK}
\end{bmatrix}\;\\
\text{and elements} \;
h^{nk}_{ml}=\frac{1}{2}\frac{\partial F_{nm} }{\partial x_{ml}}
\end{align}
where "${nk}$" refers to $n^{th}$ buyer and  $k^{th}$ resource

  \begin{theorem}\label{thm_test_stability}
        If matrix $ \left[H(\mathbf{x})+H(\mathbf{x})^{T}\right] \prec 0$ then  $F(\mathbf{x})$ is strictly monotone  and  equilibrium is unique and stable 
  \end{theorem}  

\begin{proof}
Appendix~\ref{append4}
\end{proof}
The above result provides a sufficient condition for the uniqueness and stability of the equilibrium by examining the strict monotonicity of $F(\mathbf{x})$. However, this characterization relies on the matrix $H(\mathbf{x})$ and its properties. To further our understanding, we now extend the analysis by incorporating additional structural properties of the matrix $H(\mathbf{x})$ into our framework. 
\begin{theorem}
Assume that $F(\mathbf{x})$ is continuously differentiable on $\mathcal{C}$, H is symmetric and negative semidefinite. Then there exists a real-valued concave 
function $f(\mathbf{x}): \mathcal{C} \rightarrow \mathbb{R}$ satisfying $F(\mathbf{x})=\nabla f(\mathbf{x})$,  with $\mathbf{x}^*$ the solution of VI$( F, \mathcal{C})$ also being the solution of the optimization problem:
 \begin{align}
& \underset{\mathbf{x}}{\text{Maximize}} \quad f(\mathbf{x}) \\
& \text{subject to} \quad \mathbf{x}\in \mathcal{C}
\end{align}
\end{theorem}
\begin{proof}
\cite{nagurney2009network}, CHAPTER 1, Theorem 1.1
\end{proof}
\begin{remark}
   In the previous theorem, we demonstrated that if the matrix \( H \) is symmetric and negative semi-definite, there exists a real-valued concave function \( f(\mathbf{x}) \). This function \( f(\mathbf{x}) \) can be interpreted as a potential function in the context of potential games \cite{MONDERER1996124}, implying the existence of a potential game whose equilibrium coincides with the equilibrium of $\text{VI}(\mathcal{C}, F)$. An example of this is the traditional Fisher market, where \( \sum_{n \in \mathcal{N}} B_n \log(U_n(\mathbf{x}_{n})) \) serves as a potential function.

A natural question arises: what can be concluded if matrix \( H \) is not (symmetric) negative semi definite but only the diagonal blocks \( \left(H^{n}_{n}\right)_{n \in \mathcal{N}} \) of matrix \( H \) are negative semi-definite? In this case, we can show that there exist functions \( f_{n}(\mathbf{x}_{n}, \mathbf{x}_{-n}) \), concave in \( \mathbf{x}_{n} \; \forall n \in \mathcal{N} \), such that \( \mathbf{x}^{*} \) is a VE corresponding to the equilibrium of $\text{VI}(\mathcal{C}, F)$. Specifically, \( \mathbf{x}^{*} \) solves $$\text{Maximize}_{\mathbf{x}_{n}} \quad f_{n}(\mathbf{x}_{n}, \mathbf{x}^{*}_{-n}) \quad \text{subject to} \quad \mathbf{x}_{n} \in \mathcal{C}(\mathbf{x}^{*}_{-n}) \quad \forall n \in \mathcal{N}.$$This indicates the existence of a coupled constraint game \cite{rosen1965existence} whose Nash equilibrium corresponds to the equilibrium of the $\text{VI}(\mathcal{C}, F)$. In the case of homogeneous utilities, this is the first claim we have proved in Corollary \ref{extension_corollary}. However, when the matrix \( H \) is neither symmetric nor the diagonal blocks \( \left(H^{n}_{n}\right)_{n \in \mathcal{N}} \) negative semidefinite, solving $\text{VI}(\mathcal{C}, F)$ still gives a competitive equilibrium. This shows that the variational inequality formulation provides a more general framework. Due to space constraints, we keep the detailed proof of our claims for future research.

\end{remark}


\begin{lemma}
If the utilities of the buyers in the market \(\mathcal{M}\) satisfy Assumption \ref{assum_varri}, then the aggregate excess demand function \(\mathbf{z}(\mathbf{p})\), given prices \(\mathbf{p} \in \mathbb{R}^{K}_{+}\), is strictly monotone. Specifically, for any \(\mathbf{p}^1, \mathbf{p}^2 \in \mathbb{R}^{K}_{+}\), the following inequality holds:
\[
\langle \mathbf{z}(\mathbf{p}^1) - \mathbf{z}(\mathbf{p}^2), \mathbf{p}^1 - \mathbf{p}^2 \rangle \leq 0.
\]\label{excess_demand_monoto}
\end{lemma}
\begin{proof}
Appendix~\ref{append5}
\end{proof}

\begin{theorem}
   If the buyers' utilities in market $\mathcal{M}$ satisfy Assumption 2, then $\mathbf{p}^{*}$ is an equilibrium price vector if and only if
   \begin{equation}
      \langle \mathbf{z}(\mathbf{p}), \mathbf{p} - \mathbf{p}^* \rangle \leq 0 \quad \forall \mathbf{p} \in \mathcal{P}.
   \end{equation}
\end{theorem}

\begin{proof}
Given that the buyers' utilities in market $\mathcal{M}$ satisfy Assumption 2, Lemma~\ref{excess_demand_monoto} ensures that:
\[
\langle \mathbf{z}(\mathbf{p}^1) - \mathbf{z}(\mathbf{p}^2), \mathbf{p}^1 - \mathbf{p}^2 \rangle \leq 0.
\] Then the claim follows from Theorem 3.1 \cite{dafermos1990exchange}.
\end{proof}

\section{Decentralized Learning of Competitive Equilibrium}
In the previous section, we studied the structural properties of VE within GNEP in the context of market equilibrium. Building on these properties, we now propose two learning algorithms that enable buyers to reach a CE in a decentralized manner. 

\subsection{Two-time scale stochastic approximation-based \textit{tâtonnement}}
Our first proposed algorithm is based on the \textit{tâtonnement} approach, where an auctioneer initially sets prices and buyers adjust their demands accordingly. The auctioneer then adapts prices in response to the observed demand, increasing prices when demand exceeds supply and decreasing them vice versa.

In Lemma \ref{excess_demand_monoto} of the previous section, we demonstrated that the excess demand function is monotone in price and stable at the equilibrium price. This monotonic behavior could have been sufficient to show the convergence of the \textit{tâtonnement} process:
\begin{equation}
\tag{Step\ref{Step_taton}-Algorithm \ref{algo1}}
\mathbf{p}(t+1) \leftarrow \Pi_{\mathcal{P}} \left[ \mathbf{p}(t) + \beta_{t}~\mathbf{z}(\mathbf{p}(t)) \right]
\end{equation}
provided that the game $\mathcal{G}(\mathbf{p}(t))$ is at equilibrium at each round. However, computing the NE of the game $\mathcal{G}(\mathbf{p}(t))$ in response to the prices set by the auctioneer requires an additional algorithm. Fortunately, since the game $\mathcal{G}(\mathbf{p}(t))$ is monotone, it can be solved using existing methods, such as gradient ascent \cite{Mertikopoulos2019}. Despite its theoretical soundness, this approach can be computationally intensive due to the nested loop structure: an inner loop computes the NE of $\mathcal{G}(\mathbf{p}(t))$, while the outer loop involves price adjustments by the auctioneer. To address this challenge, we propose a two-time-scale stochastic approximation-based learning algorithm. In this approach, buyers do not need to reach NE in response to each price update. Instead, they follow a gradient ascent direction with larger steps $\alpha_{t}$:
\begin{equation}
\tag{Step\ref{step2}-Algorithm\ref{algo1}}
x_{nk}(t+1) \leftarrow \left[x_{nk}(t) + \alpha_{t} \left(F_{nk}(\mathbf{x}(t)) - p_{k}(t-1)\right)\right]^{+}
\end{equation}
whereas Auctioneer adjusts prices with smaller steps $\beta_{t}$, facilitating a more efficient adjustment process and reducing the computational burden on buyers. If the steps size follow the conditions described in Algorithm 1, the algorithm converges to VE of $\text{VI}(\mathcal{C}, F)$ for details see \cite{two_time_ong2023two,two_time_bistritz2021online}

\begin{algorithm} \caption{Two-time scale stochastic approximation-based \textit{tâtonnement}}\label{algo1}
\begin{algorithmic}[1]
\Require $
  \sum_{t=0}^{\infty} \alpha_t = \sum_{t=0}^{\infty} \beta_{t}  = \infty,
  \sum_{n=0}^{\infty} \alpha_t^2 < \infty,\sum_{n=0}^{\infty}  \beta_{t}^2 < \infty,  \frac{\beta_{t}}{\alpha_t} \rightarrow 0 \quad \text{as} \quad t \rightarrow \infty.
  $
\Repeat { $t=1,2,\hdots,$}
\ForEach { $n \in \mathcal{N}$}
\ForEach{resource $k \in \mathcal{K}$}
\State Play $x_{nk}(t+1)\leftarrow  \left[x_{nk}(t) + \alpha_{t}\left(F_{nk}(\mathbf{x}(t))-p_{k}(t)\right)\right]^{+}$ \label{step2}
\EndFor
\EndFor
\State Auctioneer observes the excess demand for each resource $k\in\mathcal{K}$
$z_{k}(\mathbf{p}(t))=\left(\sum_{n}{x}^{*}_{nk}(\mathbf{p}(t))-1\right)$
\State set prices $\mathbf{p}(t+1)\leftarrow\Pi_{\mathcal{P}}\left[\mathbf{p}(t)+\beta_{t}\mathbf{z}(\mathbf{p}(t))\right]$ \label{Step_taton}
\Until {$\left \| (\mathbf{x}(t)-\mathbf{x}(t-1) \right \|\leq \epsilon$}
\end{algorithmic}
\end{algorithm}

\subsection{Trading post mechanism based Learning Algorithm }\label{sec_learning}
In the trading post mechanism, each buyer \(n\in \mathcal{N}\) places a bid \(b_{nk}\) on each type of resource \(k \in \mathcal{K}\). Once all buyers have placed their bids, the price of each resource is determined by the sum of the bids: \(p_k = \sum_{n} b_{nk}\). Each buyer \(n\) is then allocated the resource of type \(k\) in proportion to their bid. The allocation \(x_{nk}\) is given by:

\begin{equation}
x_{nk} =
\begin{cases}
\frac{b_{nk}}{p_{k}} & \text{if } b_{nk} > 0 \\ 
0 & \text{if } b_{nk} = 0 
\end{cases}
\end{equation} We assume that the buyers are price takers, and they request different amounts of the resources by submitting their bids over the resources. The auctioneer announces the prices of the resources and allocates them according to the trading post mechanism. If all buyers are satisfied with the allocation and prices announced by the auctioneer, the mechanism has reached equilibrium. Otherwise, buyers adjust their bids and resubmit them to the auctioneer. This introduces the challenge of bid dynamics: how do buyers adjust their bids to reach an equilibrium via the trading post mechanism? In this section, we devise bid updating learning schemes, which allow the buyers to reach the desired equilibrium using the trading post mechanism.
\begin{algorithm} \caption{Trading post mechanism based Learning Algorithm}\label{algo2}
\begin{algorithmic}[1]
\Require $\sum_{n=0}^{+\infty}\alpha_{t}=+\infty,\alpha_{t}\rightarrow 0 \text{ as } t\rightarrow +\infty$
\Repeat { $t=1,2,\hdots,$}
\ForEach { $n\in \mathcal{N}$}
\ForEach{resource $k \in \mathcal{K}$}
\State Play $b_{nk}(t+1)\leftarrow b_{nk}(t) + \alpha_{t} b_{nk}(t)\left(F_{nk}(\mathbf{x}(t))-p_k(t)\right)$ \label{step2_replic} 
\EndFor
\EndFor
\ForEach{$k \in \mathcal{K}$}
\State $p_{k}(t+1)\leftarrow \sum_{n\in\mathcal{N}} b_{nk}(t+1)$ \label{step_allo}
\ForEach{$n \in \mathcal{N}$}
\State $x_{nk}(t+1)\leftarrow\frac{b_{nk}(t+1)}{p_{k}(t+1)}$
\EndFor
\EndFor
\Until {$\left \| (\mathbf{p}(t)-\mathbf{p}(t-1) \right \|\leq \epsilon$}
\end{algorithmic}
\end{algorithm}In Algorithm~2 
the buyers adjust the bids \begin{equation}
\tag{Step\ref{step2_replic}-Algorithm\ref{algo2}}
   b_{nk}(t+1)\leftarrow b_{nk}(t) + \alpha_{t} b_{nk}(t)\left(F_{nk}(\mathbf{x}(t))-p_k(t)\right) 
\end{equation}
while resources are allocated  according to 
\begin{equation}
\tag{Step\ref{step_allo}-Algorithm\ref{algo2}}
    x_{nk}(t+1)\leftarrow\frac{b_{nk}(t+1)}{\sum_{n}b_{nk}(t+1)} 
\end{equation}
In Algorithm~2, Step \ref{step2_replic} and Step \ref{step_allo} collectively resemble the discrete replicator dynamics  \cite{replicator_SorinSylvain}\cite{replicator_falniowski2024discrete}, a concept from
evolutionary game theory \cite{41dcaf62-09fc-313d-abda-f1e5a405babf}. The stability of associated replicator dynamics can be ensured if the VE is stable. Leveraging this relationship, we demonstrate that the Algorithm~2 converges to the VE.

\begin{theorem}
If the function \( F \) in the problem \(\text{VI}(\mathcal{C}, F)\) satisfies Assumption~\ref{assum_varri}, then Algorithm \ref{algo2} converges to the unique VE of the problem \(\text{VI}(\mathcal{C}, F)\).
\end{theorem}
\begin{proof}
Appendix~\ref{append6}
\end{proof}
\section{Numerical Experiments}
In this section, we present numerical experiments to validate the proposed learning algorithms. We focus on a Fisher market setting where competition among buyers is modeled using Tullock contests or rent-seeking games \cite{rentseeking}. The Tullock contest framework is widely used in economics to capture competitive interactions between multiple agents. It has also found extensive applications in the communication network literature, modeling scenarios such as competition between social media users for visibility on timelines\cite{comsnet2014}, multipath TCP network utility maximization \cite{TCP}, competition among miners in multicryptocurrency blockchain networks \cite{blockchain}, and competition between service providers in communication markets to attract users \cite{Strategic_Resource_Pricing}. In the Tullock contest framework, agents expend costly resources in an effort to win a prize, with the probability of winning determined by the contest success function (CSF). The standard form of the CSF is typically expressed as
\[
\rho(x) = \frac{(x_n)^r}{\sum_{n'} (x_{n'})^r},
\]
where $x_n\in \mathbb{R}_{+}$ denotes the effort of agent $n$, and $r$ is a parameter. For example, when $r = 1$, the model describes a lottery, while $r \to \infty$ represents an all-pay auction. In our work, we generalize this model by incorporating multiple resources into the CSF. The generalized CSF is defined as 
\begin{equation}
    U_n(\mathbf{x}_n, \mathbf{x}_{-n}) = \frac{q_n(\mathbf{x}_n)}{\sum_{m \in \mathcal{N}} q_m(\mathbf{x}_m)} \quad \text{where} \quad q_n(\mathbf{x}_n) = \sum_k a_{nk} \left(x_{nk}\right)^{\rho_{nk}}\label{numerical_utiltiy}
\end{equation}
where $q_n(\mathbf{x}_n)$ represents the effort function for each agent $n$. For each buyer \(n\) and each good \(k\), the parameters satisfy \(0 < \rho_{nk} < 1\) and \(0 < a_{nk} < 1\) such that $\sum_{k}a_{nk}=1$ for each $n\in\mathcal{N}$. Given fixed resource prices, the multi-resource Tullock rent-seeking game, as defined by the utility functions in \eqref{numerical_utiltiy}, is strictly monotone and has a unique, stable variational equilibrium (see Theorem 1 \cite{Strategic_Resource_Pricing}). Moreover, the computation of the pseudo-Jacobian in \ref{append_Test_stability} confirms that $F(\mathbf{X})$, as described in \eqref{VI_F}, is strictly monotone.
\par  For numerical simulations, we consider a market scenario involving five buyers and three resources, with parameters \(\rho\) and \(a\) randomly generated. The competitive equilibrium is computed using Algorithms 1 and 2. In both algorithms, we employed step sizes defined as $\alpha_t=\frac{1}{(t+1)^{0.6}}$ and $\beta_t=\frac{1}{(t+1)^{0.9}}$ for Algorithm 1.

\newpage
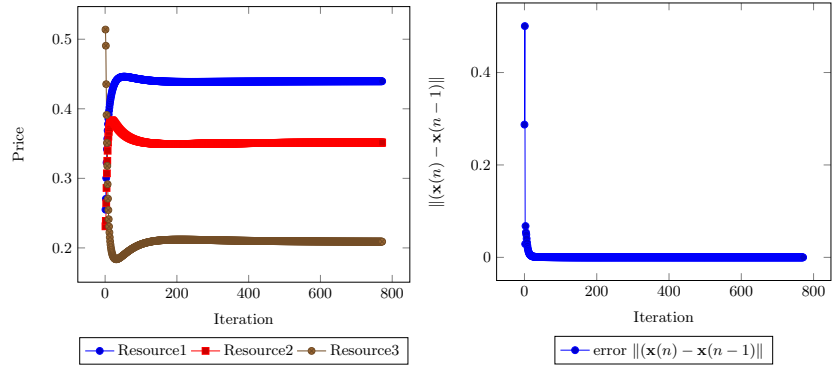
\begin{figure}
\centering
\vspace{-0.5cm}
\resizebox{0.45\textwidth}{!}
{
\begin{tikzpicture}
    \begin{axis}[
      xlabel={Iteration},
      ylabel={Price},
      legend style={at={(0.5,-0.2)}, anchor=north,legend columns=-1},
    ]

   \addplot table [x index=0, y index=1, col sep=comma] {./DATA/price_TTS_data.csv}; 
   \addplot table [x index=0, y index=2, col sep=comma] {./DATA/price_TTS_data.csv};
   \addplot table [x index=0, y index=3, col sep=comma] {./DATA/price_TTS_data.csv};
    \addlegendentry{Resource1}\addlegendentry{Resource2}\addlegendentry{Resource3}
       \end{axis}
  \end{tikzpicture}
 } 
\resizebox{0.45\textwidth}{!}
{\begin{tikzpicture}
    \begin{axis}[
      xlabel={Iteration},
      ylabel={$\left \|  (\mathbf{x}(n)-\mathbf{x}(n-1) \right \|$},
      legend style={at={(0.5,-0.2)}, anchor=north,legend columns=-1},
    ]

   \addplot table [x index=0, y index=1, col sep=comma] {./DATA/err_TTS_data.csv}; 
    \addlegendentry{error~$\left \| (\mathbf{x}(n)-\mathbf{x}(n-1) \right \|$}
    
    \end{axis}
  \end{tikzpicture}
  }
\put(-320,145){(\ref{fig:convergence_comparison_algo1}.a)}
\put(-140,145){(\ref{fig:convergence_comparison_algo1}.b)}
\caption{(a-b) Convergence two-time scale stochastic approximation-based \textit{tâtonnement}}
\label{fig:convergence_comparison_algo1}
  
    
\end{figure}
Figure~\ref{fig:convergence_comparison_algo1} (a-b) illustrates the convergence behavior of the two-time scale stochastic approximation-based \textit{tâtonnement} method for a representative instance, while Figure~\ref{fig:convergence_comparison_algo2} (a-b) describes the convergence behavior of the trading post mechanism based learning algorithm. 
\begin{figure}
\centering
\resizebox{0.45\textwidth}{!}
{
\begin{tikzpicture}
    \begin{axis}[
      xlabel={Iteration},
      ylabel={Price},
      legend style={at={(0.5,-0.2)}, anchor=north,legend columns=-1},
    ]

   \addplot table [x index=0, y index=1, col sep=comma] {./DATA/price_TP_data.csv}; 
   \addplot table [x index=0, y index=2, col sep=comma] {./DATA/price_TP_data.csv};
   \addplot table [x index=0, y index=3, col sep=comma] {./DATA/price_TP_data.csv};
    \addlegendentry{Resource1}\addlegendentry{Resource2}\addlegendentry{Resource3}
       \end{axis}
  \end{tikzpicture}
 }  
\resizebox{0.45\textwidth}{!}
{
      
    \begin{tikzpicture}
    \begin{axis}[
      xlabel={Iteration},
      ylabel={$\left \|  (\mathbf{x}(n)-\mathbf{x}(n-1) \right \|$},
      legend style={at={(0.5,-0.2)}, anchor=north,legend columns=-1},
    ]

   \addplot table [x index=0, y index=1, col sep=comma] {./DATA/err_TP_data.csv}; 
    \addlegendentry{error~$\left \| (\mathbf{x}(n)-\mathbf{x}(n-1) \right \|$}
    
    \end{axis}
  \end{tikzpicture}
  }
\put(-320,145){(\ref{fig:convergence_comparison_algo2}.a)}
\put(-140,145){(\ref{fig:convergence_comparison_algo2}.b)}
\caption{(a-b) Convergence of trading post mechanism based learning algorithm}
\label{fig:convergence_comparison_algo2}
  
    
\end{figure}
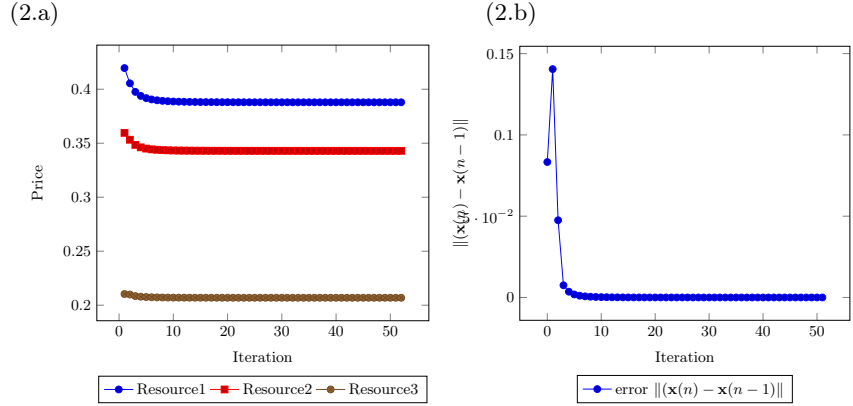
\section{Conclusion}
In this work, we have introduced a novel variational inequality formulation for competitive equilibrium problems within the generalized Fisher market model. Unlike traditional Fisher markets, where buyers' utilities depend solely on their own resource allocations, our formulation considers dependencies on the allocations of their competitors. Our proposed variational inequality framework provides a broader framework for computing competitive equilibrium in generalized Fisher markets, overcoming the constraint of homogeneous buyer utilities that is typical in traditional Fisher markets. We have examined various structural properties of this formulation, suggesting avenues for future extensions. In addition, we have developed two decentralized learning algorithms designed to facilitate buyers to achieve competitive equilibrium in a decentralized manner. Moving forward, our future research aims to explore scenarios involving buyers with unknown utility functions.


\bibliographystyle{splncs04}
\bibliography{sample-bibliography}

\appendix
\section{Appendix}
\subsection{Proof of Proposition\ref{proposition_KKT}}\label{append1}
 For each buyer $n$, the Lagrangian function of their decision problem at the generalized Nash equilibrium of the game $\widehat{\mathcal{G}}$ is given by: 
\begin{equation}
    \mathcal{L}_{n}\left ( \mathbf{x}_n,\mathbf{x}^{*}_{-n},\mu_{n},\xi_n \right )=U_{n}(\mathbf{x}_{n},\mathbf{x}_{-n}^{*})+\mu_{n}\left ( \sum_{k\in\mathcal{K}} p_{k}^{*}x_{nk}-B_n \right )+\sum_{k\in\mathcal{K}}\xi_{nk}x_{nk}
\end{equation}
KKT conditions for each buyer $n$ at a GNE $\mathbf{x}^*$
\begin{subequations}
\begin{align}
\text{(Stationarity)}\;& \left[\displaystyle \nabla _{n} \mathcal{L}_{n}\left ( \mathbf{x}_n,\mathbf{x}_{-n}^{*},\mu_{n}^{*},\xi_n^{*} \right )\right]_{\mathbf{x}_{n}=\mathbf{x}_{n}^{*}}=0\label{kkt:1}\\
\text{(Complementary slackness)}\;&\mu_n^{*} \left( \sum_{k\in\mathcal{K}}p_{k}^{*}x^{*}_{nk}=B_n\right)=0, \xi_{nk}^{*}x^{*}_{nk}=0,\; \forall k\in\mathcal{K}\label{kkt:2}\\
\text{(Primal Feasibility)}\;&\left( \sum_{k\in\mathcal{K}}p_{k}^{*}x^{*}_{nk}=B_n\right)=0\\
\text{(Dual Feasibility)}\;&\mu^{*}_{n}\geq 0,\xi^{*}_{nk}\geq 0,\; \forall k\in\mathcal{K}\label{kkt:4}
\end{align}\label{KKT_agent}
\end{subequations}

Consider the stationary condition \eqref{kkt:1}
\begin{equation}
    \left[\displaystyle \nabla_{nk} U_{n}(\mathbf{x}_n,\mathbf{x}_{-n}^{*})\right]_{\mathbf{x}_{n}=\mathbf{x}_{n}^{*}}-\mu_{n}^{*}p^{*}_{k} +\xi^{*}_{nk}=0\; \forall k\in\mathcal{K}\label{Stationarity}
 \end{equation}
Multiplying by $x^{*}_{nk}$ on both side
\begin{equation}
    \left[\displaystyle \nabla_{nk} U_{n}(\mathbf{x}_n,\mathbf{x}_{-n}^{*})\right]_{\mathbf{x}_{n}=\mathbf{x}_{n}^{*}}x^{*}_{nk}-\mu_{n}^{*}p^{*}_{k}x_{nk}^{*}+\xi^{*}_{nk}x^{*}_{nk}=0\label{before_sum}
\end{equation}
 summing over $k\in \mathcal{K}$
\begin{equation}
   \sum_{k\in\mathcal{K}}  \left[\displaystyle \nabla_{nk} U_{n}(\mathbf{x}_n,\mathbf{x}_{-n}^{*})\right]_{\mathbf{x}_{n}=\mathbf{x}_{n}^{*}}x^{*}_{nk}-\mu^{*}_{n}B_{n}=0 \label{after_sum}
\end{equation}
Where \eqref{after_sum} followed from \eqref{before_sum} as $\sum_{k}p^{*}_{k}x_{nk}^{*}=B_n$ and $\xi^{*}_{nk}x^{*}_{nk}=0$

\begin{equation}
  \mu^{*}_{n} =\frac{1}{B_n}\left[\displaystyle \nabla_{nk} U_{n}(\mathbf{x}_n,\mathbf{x}_{-n}^{*})\right]_{\mathbf{x}_{n}=\mathbf{x}_{n}^{*}}x^{*}_{nk}\label{mu_value}
\end{equation}
Replacing value of $ \mu^{*}_{n}$ in \eqref{Stationarity} from  \eqref{mu_value} and writing $\frac{\xi^{*}_{nk}}{ \mu^{*}_{n}}=\gamma^*_{nk}$

\begin{equation}
    \frac{B_n\left[\displaystyle \nabla_{nk} U_{n}(\mathbf{x}_n,\mathbf{x}_{-n}^{*})\right]_{\mathbf{x}_{n}=\mathbf{x}_{n}^{*}}}{\sum_{k'\in\mathcal{K}}\left[\displaystyle \nabla_{nk'} U_{n}(\mathbf{x}_n,\mathbf{x}_{-n}^{*})\right]_{\mathbf{x}_{n}=\mathbf{x}_{n}^{*}}x^{*}_{nk'}}-p_{k}^{*}+\gamma^{*}_{nk}=0
\end{equation}
Considering all conditions for each \( n \in \mathcal{N} \) and the decision problem of the Auctioneer, we derive the combined system of KKT conditions as presented in equation \eqref{KKT_comb}.
\subsection{Proof of Theorem 2}
If Assumption \ref{assum_varri} holds, then the first claim follows from Theorem 1.6 \cite{nagurney2009network}. Consider the following KKT system of the Variational Inequality Problem  $\text{VI}(\mathcal{C}, \mathbf{F})$. The point \( \mathbf{x}_{VE}^* \) is termed a Variational equilibrium  if it satisfies the following KKT conditions:

\begin{subequations}

\begin{align}
&\text{For each buyer $n$ in }\mathcal{N}\nonumber\\
&\text{(Stationarity)}\; \left[ \frac{B_n\displaystyle \nabla_{nk} U_{n}(\mathbf{x})}{\sum_{k'\in\mathcal{K}}\displaystyle \nabla_{nk'} U_{n}(\mathbf{x})x_{nk'}}\right]_{\mathbf{x}=\mathbf{x}^{*}_{VE}}-\lambda_{k}^{*}+\nu^{*}_{nk}=0,\;\forall k \in \mathcal{K}\label{kkt_VE:1}\\
&\text{(Complementary slackness)}\;\lambda_{k}^{*}\left(\sum_{n\in\mathcal{N}} x_{nk}^{*}-1\right)=0, \nu_{nk}^{*}x^{*}_{nk_{VE}}=0,\; \forall k\in\mathcal{K}\label{kkt_VE:2}\\
&\text{(Dual Feasibility)}\; \lambda_{k}^{*}\geq 0,\nu^{*}_{nk}\geq 0,\; \forall k\in\mathcal{K}\label{kkt_VE:3}
 \end{align}\label{KKT_VE}
\end{subequations}
If Assumptions~\ref{assump_existence_GNE} and \ref{assup_stricmono_utilti} hold, then $\mathbf{x}^*$ is a CE allocation if and only if it satisfies the KKT conditions given by \eqref{KKT_comb}. The KKT conditions \eqref{KKT_VE} and \eqref{KKT_comb} are equivalent when $\lambda_{k}^{*} = p^*_{k}$ for all $k \in \mathcal{K}$ and $\nu^{*}_{nk} = \gamma^{*}_{nk}$ for all $n \in \mathcal{N}$ and $k \in \mathcal{K}$. Since $\mathbf{x}^{*}_{VE}$ is unique, it follows that $\mathbf{x}^{*}_{VE} =\mathbf{x}^{*}$, thereby establishing the claim.
\label{append2}
\subsection{Proof of Corollary 1}\label{append3}
From Proposition~\ref{proposition_KKT}, the pair \((\mathbf{p}^*,\mathbf{x}^*)\) is a competitive equilibrium for the market \(\mathcal{M}\) if it satisfies the system of KKT conditions \eqref{KKT_comb}. Consider the stationarity condition from the KKT conditions \eqref{kkt_comb:1}:

\begin{align}
 \left[ \frac{B_n \nabla_{nk} U_{n}(\mathbf{x}_n, \mathbf{x}_{-n}^{*})}{\sum_{k' \in \mathcal{K}} \nabla_{nk'} U_{n}(\mathbf{x}_n, \mathbf{x}_{-n}^{*}) x_{nk'}} \right]_{\mathbf{x}_{n} = \mathbf{x}_{n}^{*}} - p_{k}^{*} + \gamma^{*}_{nk} = 0, \quad \forall k \in \mathcal{K}
\end{align}

If the utility function \(U_n\) of each buyer \(n \in \mathcal{N}\) is homogeneous of degree 1, then by Euler's Homogeneous Function Theorem, we have:

\begin{align}
\sum_{k' \in \mathcal{K}} \nabla_{nk'} U_{n}(\mathbf{x}_n, \mathbf{x}_{-n}^{*}) x_{nk'} = U_{n}(\mathbf{x}_n, \mathbf{x}_{-n}^{*})
\end{align}
Substituting this result into the stationarity condition simplifies it to:
\begin{align}
\left[ \frac{B_n \nabla_{nk} U_{n}(\mathbf{x}_n, \mathbf{x}_{-n}^{*})}{U_{n}(\mathbf{x}_n, \mathbf{x}_{-n}^{*})} \right]_{\mathbf{x}_{n} = \mathbf{x}_{n}^{*}} - p_{k}^{*} + \gamma^{*}_{nk} = 0, \quad \forall k \in \mathcal{K}
\end{align}
This expression can be further simplified by recognizing that \(\frac{\nabla_{nk} U_{n}(\mathbf{x}_n, \mathbf{x}_{-n}^{*})}{U_{n}(\mathbf{x}_n, \mathbf{x}_{-n}^{*})}\) is the derivative of the logarithm of the utility function with respect to \(x_{nk}\):
\begin{align}
\left[ \nabla_{nk} B_n\left(\log U_{n}(\mathbf{x}_n, \mathbf{x}_{-n}^{*})\right) \right]_{\mathbf{x}_{n} = \mathbf{x}_{n}^{*}} - p_{k}^{*} + \gamma^{*}_{nk} = 0, \quad \forall k \in \mathcal{K}
\end{align}
Now  consider the coupled constraint game between the buyers, let us examine the decision problem for each buyer \( n \), which is formulated as follows:

\begin{subequations}\label{galeconvex}
\begin{align}
&\hspace{5cm} \underset{{\mathbf{x}_n}}{\text{Maximize :}} 
& & B_n\log(U_n(\mathbf{x}_{n},\mathbf{x}_{-n})) \label{primal}\\
& \hspace{5cm} \text{subject to} 
& & \sum_{n \in \mathcal{N}} x_{nk} \leq 1, \quad \forall k \in \mathcal{K} \label{const3}.
\end{align}
\end{subequations}
Here, \(\mathbf{x}^*\) represents a varriational equilibrium of the game if and only if, for each buyer \( n \), the best response \(\mathbf{x}^{*}_{n}\) given the strategies of the opponents \(\mathbf{x}^{*}_{-n}\) satisfies the KKT conditions. These conditions are stated as follows:
\begin{subequations}
\begin{align}
&\text{(Stationarity)}\;\left[\nabla_{nk}\left(B_n\log{U(\mathbf{x}_{n},\mathbf{x}_{-n}^{*})}\right)\right]_{\mathbf{x}_{n}^*} - \lambda_{k}^{*} + \gamma_{nk}^* = 0, \quad \forall k \in \mathcal{K} \\
&\text{(Complementary Slackness)}\; \lambda_{k}^{*} \left(\sum_{n \in \mathcal{N}} x_{nk}^{*} - 1\right) = 0, \quad \gamma_{nk}^{*} x^{*}_{nk} = 0 \\
&\text{(Dual Feasibility)}\;\lambda^*_{k} \geq 0, \quad \gamma_{nk}^{*} \geq 0, \quad \forall k \in \mathcal{K}
\end{align}
\label{KKT_agent1}
\end{subequations}
These conditions ensure that \(\mathbf{x}^*\) satisfies both the stationarity, complementary slackness, and dual feasibility requirements, thus proving the first claim.

\par Similarly, in the second scenario, where each buyer's utility depends solely on their own decision variable, the KKT conditions at the optimal solution of the EG-Program \ref{EG_prog} precisely align with the combined KKT conditions \eqref{KKT_comb}. This alignment validates the second claim.

\subsection{Proof of Theorem \ref{thm_test_stability}}\label{append4}
 If the matrix \( \left[ H(\mathbf{x}) + H(\mathbf{x})^{T} \right] \prec 0 \), then \( F(\mathbf{x}) \) is strictly monotone. The proof follows the same steps as those in Theorem 6 of \cite{rosen1965existence}.
 Consider two distinct points $\mathbf{x}^{1}$ and $\mathbf{x}^{2}$ and $\mathbf{x}(\theta)=\theta \mathbf{x}^{1}+(1-\theta)\mathbf{x}^{2}$
    \begin{equation}
        \frac{ F(\mathbf{x}(\theta))}{d\theta}=H(\mathbf{x}(\theta)) \frac{ \mathbf{x}(\theta)}{d\theta}
    \end{equation}
    \begin{equation}
        \frac{ F(\mathbf{x}(\theta))}{d\theta}=H(\mathbf{x}(\theta))(\mathbf{x}^2-\mathbf{x}^1) 
    \end{equation}
      \begin{equation}
        \left[F(\mathbf{x}^2)-F(\mathbf{x}^1)\right]= \int_{0}^{1}H(\mathbf{x}(\theta))(\mathbf{x}^2-\mathbf{x}^1) d\theta
    \end{equation}
    Multiplying on both side by $(\mathbf{x}^2-\mathbf{x}^1)$
     \begin{equation}
       (\mathbf{x}^2-\mathbf{x}^1)\left[F(\mathbf{x}^2)-F(\mathbf{x}^1)\right]=\int_{0}^{1}(\mathbf{x}^2-\mathbf{x}^1)H(\mathbf{x}(\theta))(\mathbf{x}^2-\mathbf{x}^1) d\theta
    \end{equation}
       \begin{equation}
       (\mathbf{x}^2-\mathbf{x}^1)\left[F(\mathbf{x}^2)-F(\mathbf{x}^1)\right]=\frac{1}{2} \int_{0}^{1}(\mathbf{x}^2-\mathbf{x}^1)\left[H(\mathbf{x}(\theta))+H'(\mathbf{x}(\theta))\right](\mathbf{x}^2-\mathbf{x}^1) d\theta
    \end{equation}
  Furthermore, if \( F(\mathbf{x}) \) is strictly monotone, then uniqueness and global stability follow from Proposition 3.
\subsection{Proof of Lemma 1}\label{append5}
Given prices \(\mathbf{p}^1\) and \(\mathbf{p}^2\), the demands \(\mathbf{x}^{*}(\mathbf{p}^1)\) and \(\mathbf{x}^{*}(\mathbf{p}^2)\) are NE of the game $\mathcal{G}(\mathbf{p}^1)$ and $\mathcal{G}(\mathbf{p}^2)$ repestively. Since \(\mathbf{x}^{*}(\mathbf{p}^1)\) is an NE, the corresponding KKT conditions given by \eqref{KKT_agent} must be satisfied for \(\mathbf{p}\)=\(\mathbf{p}^1\). Let us specifically consider the first-order stationarity condition \eqref{kkt:1} from these KKT conditions.

  \begin{equation}
    \frac{B_n\displaystyle \nabla_{nk}U_ n(\mathbf{x}^{*}(\mathbf{p}^1))}{\sum_{k'\in\mathcal{K}}\displaystyle \nabla_{nk'}U_n(\mathbf{x}^{*}(\mathbf{p}^1))x^{*}_{nk'}(\mathbf{p}^1)}+\xi^{1}_{nk}=p^{1}_k,\;  \forall n\in \mathcal{N}, \forall k\in \mathcal{K}\label{exces_dem1}
\end{equation}

    Similarly, for NE  $\mathbf{x}^{*}(\mathbf{p}^2)$ given prices $\mathbf{p}^2$ 
\begin{equation}
    \frac{B_n\displaystyle \nabla_{nk}U_n(\mathbf{x}^{*}(\mathbf{p}^2))}{\sum_{k'\in\mathcal{K}}\displaystyle \nabla_{nk'}U_n(\mathbf{x}^{*}(\mathbf{p}^2))x^{*}_{nk'}(\mathbf{p}^2)}+\xi^{2}_{nk}=p^{2}_k,\; \forall n\in \mathcal{N}, \forall k\in \mathcal{K} \label{exces_dem2}
\end{equation}
Subtracting  \eqref{exces_dem2} from  \eqref{exces_dem1} and multiplying by  $\left({x^{*}_{nk}}(\mathbf{p}^1)-x^{*}_{nk}(\mathbf{p}^2)\right)$ on both side 
 \begin{dmath}
        \left(\frac{B_n\displaystyle \nabla_{nk}U_n(\mathbf{x}^{*}(\mathbf{p}^1))}{\sum_{k\in\mathcal{K}}\displaystyle \nabla_{nk}U_n(\mathbf{x}^{*}(\mathbf{p}^1))x^{*}_{nk}(\mathbf{p}^1)}- \frac{B_n\displaystyle \nabla_{nk}U_n(\mathbf{x}^{*}(\mathbf{p}^2))}{\sum_{k\in\mathcal{K}}\displaystyle \nabla_{nk}U_n(\mathbf{x}^{*}(\mathbf{p}^2))x^{*}_{nk}(\mathbf{p}^2)}\right)\left({x^{*}_{nk}}(\mathbf{p}^1)-x^{*}_{nk}(\mathbf{p}^2)\right)
+\left(\xi^{1}_{nk}-\xi^{1}_{nk}\right)\left({x^{*}_{nk}}(\mathbf{p}^1)-x^{*}_{nk}(\mathbf{p}^2)\right)=\left(p^{1}_k-p^{1}_k\right)\left({x^{*}_{nk}}(\mathbf{p}^1)-x^{*}_{nk}(\mathbf{p}^2)\right)
    \end{dmath}
Now, summing over for all $k\in \mathcal{K}$ and for all $n\in \mathcal{N}$
\begin{dmath}
       \sum_{n}\sum_{k} \left(\frac{B_n\displaystyle \nabla_{nk}U_n(\mathbf{x}^{*}(\mathbf{p}^1))}{\sum_{k\in\mathcal{K}}\displaystyle \nabla_{nk}U_n(\mathbf{x}^{*}(\mathbf{p}^1))x^{*}_{nk}(\mathbf{p}^1)}- \frac{B_n\displaystyle \nabla_{nk}U_n(\mathbf{x}^{*}(\mathbf{p}^2))}{\sum_{k\in\mathcal{K}}\displaystyle \nabla_{nk}U_n(\mathbf{x}^{*}(\mathbf{p}^2))x^{*}_{nk}(\mathbf{p}^2)}\right)\left({x^{*}_{nk}}(\mathbf{p}^1)-x^{*}_{nk}(\mathbf{p}^2)\right)-\xi^{1}_{nk}{x^{*}_{nk}}(\mathbf{p}^2)-\xi^{2}_{nk}{x^{*}_{nk}}(\mathbf{p}^1)=  \sum_{n}\sum_{k}\left(p^{1}_k-p^{1}_k\right)\left({x^{*}_{nk}}(\mathbf{p}^1)-x^{*}_{nk}(\mathbf{p}^2)\right)
\end{dmath}
The left-hand side of the above equation is negative due to Assumption\ref{assum_varri}, and for the right-hand side 
\begin{align}
\sum_{n}\sum_{k}\left(p^{1}_k-p^{1}_k\right)\left({x^{*}_{nk}}(\mathbf{p}^1)-x^{*}_{nk}(\mathbf{p}^2)\right)   = &\sum_{k} \left(\sum_{n}{x^{*}_{nk}}(\mathbf{p}^1)-x^{*}_{nk}(\mathbf{p}^2)\right)(p^{1}_k- p^2_k)\\
 = &\sum_{k} \left(z_{k}(\mathbf{p}^2)-z_{k}(\mathbf{p}^2)\right)(p^{1}_k- p^2_k)
    \end{align}
which proves the 

\[
\langle \mathbf{z}(\mathbf{p}^1) - \mathbf{z}(\mathbf{p}^2), \mathbf{p}^1 - \mathbf{p}^2 \rangle \leq 0.
\]
\subsection{Proof of Theorem 5}\label{append6}
To show the convergence of Algorithm \ref{algo2}, we first consider the continuous version of the algorithm (as in \cite{uncoupled}) represented by an ordinary differential equation (ODE) and show that the dynamical system is Lyapunov stable. Then, applying the stochastic approximation technique from \cite{Stochastic}, we achieve the desired results. Consider the continuous version of step \eqref{step2} in Algorithm~2, represented by the following ODE:
\begin{align}
  \dot{b}_{nk} = b_{nk}(s) \left( F_{nk}(\mathbf{x}(s)) - p_k(s) \right).
\end{align}
After factoring out \( p_k(s) \) and multiplying by \( b_{nk}(s) \), we obtain:
\begin{align}
  \dot{b}_{nk} = p_k(s) \left( F_{nk}(\mathbf{x}(s)) x_{nk}(s) - b_{nk}(s) \right). \label{ode1}
\end{align}
Now consider \( x_{nk} = \frac{b_{nk}}{p_k} \). Calculating the time $(s)$ derivative of \( x_{nk} \), we get:
\begin{align}
  \dot{x}_{nk}(s) &= \frac{p_{k}(s) - b_{nk}(s)}{(p_{k}(s))^2} \dot{b}_{nk}(s) - b_{nk}(s) \sum_{m \neq n, m\in \mathcal{N}} \frac{\dot{b}_{mk}(t)}{(p_{k}(s))^2} \\
  \dot{x}_{nk}(s) &= \frac{\dot{b}_{nk}(s)}{p_{k}(s)} - \frac{b_{nk}(t)}{p_{k}(s)^2} \sum_{m} \dot{b}_{mk}(s) \\
  \dot{x}_{nk}(s) &= \frac{1}{p_k(s)} \left[ \dot{b}_{nk}(s) - x_{nk}(s) \sum_{m} \dot{b}_{mk}(s) \right].
\end{align}
Substituting \(\dot{b}_{nk}\) from \eqref{ode1}, for all \(n \in \mathcal{N}\) and \(k \in \mathcal{K}\), we get:
\begin{dmath}
  \dot{x}_{nk}(s) = x_{mk}(s) \left[ F_{nk}(\mathbf{x}(s)) - \sum_{m \in \mathcal{N}} F_{mk}(\mathbf{x}(s)) x_{mk}(s) \right],
\end{dmath}
where we used \(x_{nk}(s) \left(\sum_{n} b_{nk}(s)\right) = x_{nk}(s) \left(p_k(s)\right) = b_{nk}\).

To show the stability of the above dynamics, consider a Lyapunov function
\[
V(\mathbf{x}) = \sum_{n \in \mathcal{N}} \sum_{k \in \mathcal{K}} x^{*}_{nk} \ln{\frac{x^{*}_{nk}}{x_{nk}}}.
\]
Taking the derivative with respect to time $s$:
\begin{align}
  \dot{V}(\mathbf{x}) = \sum_{n \in \mathcal{N}} \sum_{k \in \mathcal{K}} F_{nk}(\mathbf{x})(x_{nk} - x^{*}_{nk}).
\end{align}
As \( \mathbf{x}^{*} \) is a variational equilibrium is unique and variationally stable.
\begin{align}
\sum_{n\in \mathcal{N}}\sum_{k\in \mathcal{K}}F_{nk}(\mathbf{x})(x_{nk}-x^{*}_{nk})<0\label{negt}
\end{align}
For a discrete-time setting, writing $V_{t}=V(\mathbf{x}(s))$ and then writing the first-order Taylor expansion gives
\begin{dmath}
V_{t+1}=V_{t}+\alpha_{n}\left ( \sum_{n\in \mathcal{N}}\sum_{k\in \mathcal{K}}\left (  x_{nk} -{{x}_{nk}}^{*}  \right )F_{nk}(\mathbf{x})\right )+\mathcal{O}(\alpha^2_{t})
\end{dmath} 
As in our case $F_{n}(\mathbf{x}_{n},\mathbf{x}_{-n})$ is bounded, $\mathcal{O}(\alpha^2_{t})$ is uniformly bounded by some $\frac{1}{2}M\alpha_{t}^{2}$. Now let us consider that iteration stays a bounded distance away from $\left ( \mathbf{x}^*\right )$ by some distance $c$ then from $\eqref{negt}$
\begin{dmath}
V_{t+1}\leq V_{t}-\alpha_{t}c+\mathcal{O}((\alpha_{t})^2)
\end{dmath}
Writing $V_{t}$ in terms $V_{r}$ $\forall r$ from $0\hdots t$ recursively gives  
\begin{dmath}
V_{t}\leq V_{0}-c\sum_{r=0}^{r=t}\alpha_{s}+\mathcal{O}((\alpha_t)^2) 
\end{dmath}
as $t \rightarrow \infty, V_{t} \rightarrow -\infty$ because of assumption $\sum_{t=1}^{\infty}\alpha_{t}=\infty$ and $\sum_{t=1}^{\infty}\alpha^{2}_{t}< \infty  $ which contradicts the definition of a Lyapunov function that $V_{t}>0 $. Hence $(\mathbf{x}(t))$ comes arbitrarily close to $(\mathbf{x}^{*}(t))$ infinitely often. Then convergence of $\left ( \mathbf{x}(t)\right )$ to the unique VE follows from Theorem 6.9 in \cite{{Stochastic}} which completes the proof. 
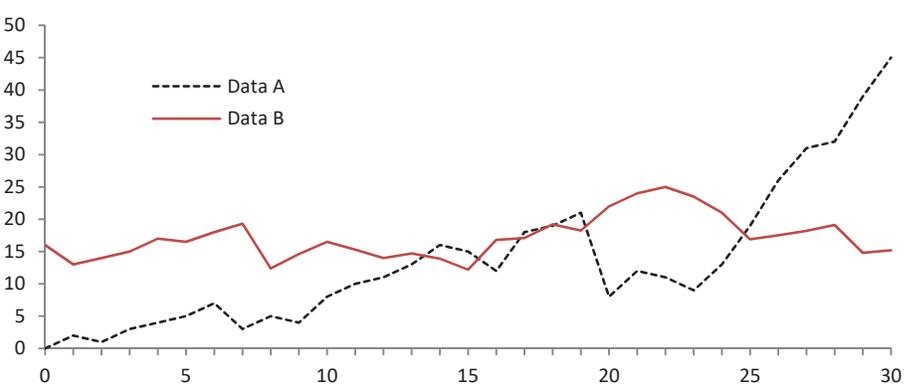
\end{document}